\documentclass[a4letter, 10pt, journal]{IEEEtran}\usepackage{amsmath,amssymb,textcomp} 
\usepackage{graphicx,mathrsfs}
\usepackage{cases}
\usepackage{multirow}
\usepackage[braket,qm]{qcircuit}
\usepackage{dblfloatfix}
\usepackage{mathtools}

\newtheorem{theorem}{Theorem}

\newtheorem{corollary}[theorem]{Corollary}

\newtheorem{definition}[theorem]{Definition}

\newtheorem{lemma}[theorem]{Lemma}

\newtheorem{remark}[theorem]{Remark}

\newtheorem{assumption}[theorem]{Assumption}
\newenvironment{proof}[1][Proof]{\textbf{#1.} }{\ \rule{0.5em}{0.5em}}

\begin{document}

\title{Nonlinear Autoregression with Convergent Dynamics on Novel Computational Platforms}

\author{Jiayin Chen and Hendra I. Nurdin\thanks{J. Chen and H. I. Nurdin are with the School of Electrical Engineering and 
Telecommunications,  UNSW Australia,  Sydney NSW 2052, Australia (\texttt{email: jiayin.chen@unsw.edu.au,h.nurdin@unsw.edu.au})}}


\maketitle

\begin{abstract}
Nonlinear stochastic modeling is useful for describing complex engineering systems. Meanwhile, neuromorphic (brain-inspired) computing paradigms are developing to tackle tasks that are challenging and resource intensive on digital computers. An emerging scheme is reservoir computing which exploits nonlinear dynamical systems for temporal information processing. This paper introduces reservoir computers with output feedback as stationary and ergodic infinite-order nonlinear autoregressive models. We highlight the versatility of this approach by employing classical and quantum reservoir computers to model synthetic and real data sets, further exploring their potential for control applications.
\end{abstract}


\section{Introduction}
\label{sec:intro}
The on-going quest for modeling complex systems has motivated a fruitful development in nonlinear stochastic modeling \cite{fan2008nonlinear,ljung2010perspectives}. The threshold model has been applied for ecology and hydrology times series and the autoregressive conditional heteroscedastic model is useful for volatility clustering \cite{fan2008nonlinear}. Well-known nonlinear system identification models include the Volterra series \cite{boyd1985fading}, neural networks \cite{kumpati1990identification}, nonlinear autoregressive exogenous models \cite{leontaritis1985input} and block-oriented models \cite{schoukens2017identification}. 

To solve tasks that are costly on digital computers, neuromorphic computing imitates human learning with the energy efficiency of the human brain. An emerging neuromorphic scheme is reservoir computing (RC) \cite{tanaka2019recent,nakajima2021reservoir}, which exploits nonlinear dynamics (the ``reservoir'') to process time-varying input signals. In this work, we establish a  theoretical framework for using reservoir computers (also abbreviated as RCs) as \textit{infinite-order nonlinear autoregressive models with exogenous inputs, or NARX($\infty$) models} for applications such as time series modelling and system identification. Such NARX($\infty$) can be also expressed as \textit{infinite-order nonlinear moving average models with exogenous inputs, or NMAX($\infty$) models.} Exploiting this equivalence, we show that such NARX($\infty$) models are asymptotically stationary and ergodic in the sense of Birkhoff-Khinchin \cite[Theorem~24.1]{billingsley2008probability}.

The attractiveness of our approach is that any nonlinear reservoir dynamics with the convergence property can induce stationary and ergodic NARX($\infty$). The reservoir dynamics are often randomly chosen but fixed at the outset and do not require precise tuning, only a linear output function is optimized to approximate target outputs. Echo-state networks (ESNs), a pioneering software RC implementation, has demonstrated remarkable performance in chaotic system modeling \cite{pathak2018model} and time series modeling \cite{kim2020time}. Their training efficiency enables RCs for fast signal processing. An FPGA RC reached 160 MHz rate for chaotic dynamics prediction and a photonic RC classified speech at a million words per second. RCs are also used in edge computing to reduce computation and transmission overhead; see \cite{tanaka2019recent,nakajima2021reservoir} and references therein. Quantum reservoir computers (QRCs) have recently been proposed to harness nonlinear quantum dynamics \cite{fujii2017harnessing, chen2019towards, chen2020temporal}. The micro-second quantum dynamics and their low energy requirements make them of interest as RC hardware \cite{markovic2020quantum}.

Central to our development is the convergence property \cite{pavlov2008convergent,tran2018convergence} of a dynamical system, also known as the echo-state property \cite{grigoryeva2018echo} in the RC literature. Roughly speaking, a convergent dynamic has a unique reference state solution defined and bounded both backwards and forwards in time. All other solutions asymptotically converge to the reference state solution, independent of their initial conditions.

References \cite{grigoryeva2016reservoir,gonon2019reservoir} investigate RCs for forecasting, reconstruction and filtering under stationary inputs. Ref.~\cite{bollt2021explaining} explores RC properties that make them effective in stochastic modeling tasks. Different from these previous works, we develop a general theory for realizing NARX($\infty$) models with RCs, taking into account the stability of the model, and conditions for the asymptotic stationarity and ergodicity. Ref.~\cite{grigoryeva2016reservoir} also considers optimizing the reservoir parameters to maximize the RC's information processing capacity via a Taylor expansion. Although we do not consider reservoir design problem here, it will be an interesting future research theme continuing from this work.


To highlight the versatility of our proposal, we employ ESNs and QRCs to model  data sets collected from diverse fields of interest. 
We cast parameter estimation for these RCs as convex optimization problems. 
Numerical experiments indicate that ESNs and QRCs with only a few tunable parameters are able to describe these data sets.

This paper is organized as follows. Sec.~\ref{sec:rc} introduces RCs and the convergence property. Sec.~\ref{sec:NARX} presents NARX($\infty$) models realized by convergent RCs and establishes their stationarity and ergodicity. Sec.~\ref{sec:para_est} details parameter estimation for ESNs and QRCs as NARX($\infty$). Sec.~\ref{sec:ts} presents numerical experiments. A conclusion is presented in Sec.~\ref{sec:conclu}.

\textbf{Notations}: $\mathbb{R}$ ($\mathbb{Z}$) are reals (integers), $\mathbb{Z}_{-} = \{ \ldots, -1, 0\}$ and $\mathbb{Z}_{+} = \{0, 1, \ldots\}$. $ x \in \mathbb{R}^n$ is an $n \times 1$ vector and $x^\top$ is its transpose. $(\mathbb{R}^n)^{\mathbb{Z}}$ is the set of infinite sequences, i.e., $u \in (\mathbb{R}^n)^{\mathbb{Z}}$ with $u = \{u_k\}_{k \in \mathbb{Z}}$ and $u_k \in \mathbb{R}^n$. $(\mathbb{R}^n)^{\mathbb{Z}_{-}}$ is the set of left-infinite sequences. $P^{\mathbb{Z}_{-}}_n: (\mathbb{R}^{n})^\mathbb{Z} \rightarrow (\mathbb{R}^n)^{\mathbb{Z}_{-}}$ is the projection. For any $\tau \in \mathbb{Z}$, $z^{-\tau}_n$ is the time shift operator, i.e., for any $u \in (\mathbb{R}^n)^{\mathbb{Z}}$ and $k \in \mathbb{Z}$, $z^{-\tau}_n(u) \rvert_k \coloneqq u_{k-\tau}$ (we drop the subscript and superscript $n$ when $n=1$). $\| \cdot \|$ is the Euclidean norm.

\section{Reservoir computing}
\label{sec:rc}
We consider RCs described by reservoir dynamics (also called the activation function in machine learning) $f: \mathbb{R}^N \times \mathbb{R}^n \times \mathbb{R} \rightarrow \mathbb{R}^N$ and output function $h: \mathbb{R}^N \rightarrow \mathbb{R}$, for all $k \in \mathbb{Z}$
\begin{equation} \label{eq:ol-rc}
\begin{cases}
x_{k+1}  = f(x_{k}, u_{k}, e_{k}), \\
\hspace*{1.1em} \hat{y}_{k} = h(x_{k}),
\end{cases}
\end{equation}
Here, $x_k \in \mathbb{R}^N$ is the state, $u_k \in \mathbb{R}^n$ is the input, $e_k \in \mathbb{R}$ models an external noise, and $\hat{y}_k \in \mathbb{R}$ is the RC output.
The dynamics $f = f_\gamma$ is often parametrized by a parameter $\gamma \in \mathbb{R}^{p}$ chosen according to the task. Here, we consider an arbitrary but fixed $f$ by uniformly randomly choosing $\gamma$ and fixing it at the onset. Only $h$ is optimized to match the target outputs. More generally, $\gamma$ can be optimized according to some criterion, a problem known as the reservoir design problem \cite{grigoryeva2016reservoir}, as alluded to previously.

\subsection{The convergence property} 
\label{sec:cv}
RCs described by \eqref{eq:ol-rc} with the convergence property \cite{pavlov2008convergent,tran2018convergence} induce input-output maps, mapping from $u \in (\mathbb{R}^n)^{\mathbb{Z}}, e \in \mathbb{R}^{\mathbb{Z}}$ to $\hat{y} \in \mathbb{R}^{\mathbb{Z}}$. Let $\phi(k; k_0, \xi)$ be a solution to \eqref{eq:ol-rc} parameterized by $u$ and $e$, starting at time $k_0$ with initial condition $x_{k_0} = \xi $. That is, for all $k \geq k_0$, $\phi(k+1; k_0, \xi) = f(\phi(k; k_0, \xi), u_k, e_k)$ and $\phi(k_0; k_0, \xi) = \xi$. A function $\beta: [0, \infty) \times \mathbb{Z}_{+} \rightarrow \mathbb{R}$ is $\mathcal{KL}$ if $\beta(0, \cdot)=0$, continuous and strictly increasing in the first argument, non-increasing in the second argument with $\lim_{t \rightarrow \infty} \beta(s, t) = 0 $ for all $s \in [0, \infty)$ \cite{lin1996smooth}. 

\begin{definition}[Convergence property \cite{pavlov2008convergent,tran2018convergence}] \label{defn:cv}
An RC described by \eqref{eq:ol-rc} has the convergence property (or is convergent) if for any $u \in (\mathbb{R}^n)^{\mathbb{Z}}$ and $e\in \mathbb{R}^\mathbb{Z}$,
\begin{itemize}
 \item[(i)] there exists a unique and bounded solution $x^{*} \in (\mathbb{R}^{N})^{\mathbb{Z}}$ to \eqref{eq:ol-rc} that satisfies $x^*_{k+1}=f(x^*_k, u_k ,e_k)$ for all $k \in \mathbb{Z}$ and $\sup_{k \in \mathbb{Z}}\| x^*_k \| < \infty$;
 \item[(ii)] there exists $\beta \in \mathcal{KL}$ (independent of $u, e$) such that, for any $k, k_0 \in \mathbb{Z}$ with $k \geq k_0$ and any $ \xi \in \mathbb{R}^{N}$,
\begin{equation} \label{eq:defn-cv}
\| x^*_k - \phi(k; k_0, \xi) \| \leq \beta(\| x^*_{k_0} - \xi \|, k-k_0).
\end{equation} 
\end{itemize}
\end{definition}
Note that the convergence property is a property of $f$.
The unique and bounded solution $x^*$ in Definition~\ref{defn:cv} is the reference state solution (determined by $u, e$ and $f$). Equation \eqref{eq:defn-cv} imposes that as $k_0 \rightarrow -\infty$, any solution $\phi(k; k_0, \xi)$ to \eqref{eq:ol-rc} asymptotically converges to the reference state solution $x^*$, independent of its initial condition $\xi$. 

The following theorem provides sufficient conditions to ensure that \eqref{eq:ol-rc} is convergent, which will be employed in parameter estimation of NARX($\infty$) models realized by convergent feedback RC dynamics; see Sec.~\ref{sec:para_est} below.

\begin{theorem} \cite[Theorem~13]{tran2018convergence}  \label{theorem:cv}
An RC with a compact state-space described by \eqref{eq:ol-rc} is convergent if there exists some $P = P^\top, P >0$ and some $\theta \in (0, 1)$ (independent of $u$ and $e$) such that for any $u_k \in \mathbb{R}^n$, $e_k \in \mathbb{R}$ and any $x_1, x_2 \in \mathbb{R}^{N}$,
\begin{equation} \label{eq:theorem-cv1}
\| f(x_1, u_k, e_k) - f(x_2, u_k, e_k) \|_P \leq \theta \| x_1 - x_2 \|_P,
\end{equation}
where $\| x \|_P \coloneqq \sqrt{x^\top P x}$.
\end{theorem}

\subsection{Filters and functionals}
\label{subsec:filters-functionals}
If an RC governed by \eqref{eq:ol-rc} is convergent, then it induces a unique time-invariant and causal filter $U_{f, h} : (\mathbb{R}^n)^{\mathbb{Z}} \times \mathbb{R}^\mathbb{Z} \rightarrow \mathbb{R}^{\mathbb{Z}}$ such that when evaluated at any time $k \in \mathbb{Z}$, $\hat{y}_k =  U_{f, h}(u, e)\rvert_{k} \coloneqq h(x^*_k)$, where $x^*$ is the reference state solution to \eqref{eq:ol-rc} (see \cite{grigoryeva2018echo,chen2019learning}). There is a bijection between $U_{f, h}$ and its associated functional $F_{f, h}: (\mathbb{R}^n)^{\mathbb{Z}_{-}} \times \mathbb{R}^{\mathbb{Z}_{-}} \rightarrow \mathbb{R}$, defined as $F_{f, h}(u', e') \coloneqq U_{f, h}(\tilde{u}', \tilde{e}')\rvert_{0}$ \cite{boyd1985fading}. Here $\tilde{u}', \tilde{e}'$ are arbitrary extensions of $u', e'$ to $(\mathbb{R}^n)^{\mathbb{Z}}, \mathbb{R}^\mathbb{Z}$. We can recover $U_{f,h}$ from $F_{f, h}$ via $U_{f, h}(u, e)\rvert_{k} =  F_{f, h}  (P_{n}^{\mathbb{Z}_{-}} \circ z_{n}^{-k}(u), P^{\mathbb{Z}_{-}} \circ z^{-k}(e))$ for any $ k \in \mathbb{Z}$. This bijection allows us to establish the measurability of $U_{f, h}$ by showing that $F_{f,h}$ is measurable in Lemma~\ref{lemma:narx}, and the Birkhoff-Khinchin ergodicity of $U_{f, h}$ by showing that $F_{f, h}$ is integrable in Lemma~~\ref{lemma:stationary}.

\section{NARX($\infty$) models by convergent dynamics}
\label{sec:NARX}
We are interested in implementing NARX($\infty$) models using RCs with output-feedback (also see Fig.~\ref{fig:cl_rc}), for all $k \in \mathbb{Z}$
\begin{equation} \label{eq:cl-rc}
\begin{cases}
x_{k+1} = g(x_k, u_k, y_k), \\
\hspace*{1.1em} \hat{y}_k = h(x_k),
\end{cases}
\end{equation}
where $x_k \in \mathbb{R}^N$ is the state and $u \in (\mathbb{R}^n)^\mathbb{Z}$ is the input. The target output $y_k \in \mathbb{R}$ is related to the one-step ahead prediction $\hat{y}_k$ via $y_k = \hat{y}_k + e_k$, where $e_k \in \mathbb{R}$ is a noise source for the model. Later on, we will consider $e$ and $u$ modeled by discrete-time stochastic processes. We also consider an equivalent representation of \eqref{eq:cl-rc} given by \eqref{eq:ol-rc}
where $u$ and $e$ are viewed as external inputs and
\begin{equation} \label{eq:f}
f(x_k, u_k, e_k) \coloneqq g(x_k, u_k, h(x_k) +e_k).
\end{equation}
We say that~\eqref{eq:cl-rc} is convergent if and only if \eqref{eq:ol-rc} is convergent.
In system identification, $u$ is often designed to be persistently exciting \cite{ljung2010perspectives} to excite all modes of the plant. 
In time series modeling, RCs make (one-step ahead) predictions $\hat{y}_k$ of $y_k$ based on a single sample path $y$ from the generating model, and are not driven by input $u$. That is, \eqref{eq:cl-rc} becomes
\begin{equation*}
\begin{cases}
x_{k+1} = g(x_k, y_k) , \\ 
\hat{y}_k = h(x_k).
\end{cases}
\end{equation*}
\begin{figure}[!ht]
\centering
\includegraphics[scale=0.8]{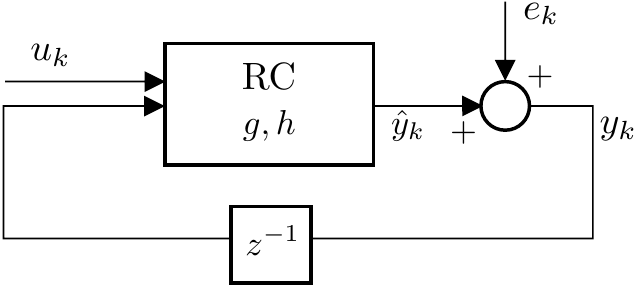}
\caption{Schematic of RCs operating in an output-feedback configuration described by \eqref{eq:cl-rc}, where $z^{-1}$ is the one-step time delay operator.}
\label{fig:cl_rc}
\end{figure}

This section introduces a probabilistic framework and shows that RCs described by \eqref{eq:cl-rc} (under certain conditions on $f$ and $h$) implement NARX($\infty$) models defined later in \eqref{eq:y-narx}. We then show that outputs of the NARX($\infty$) models are stationary and ergodic. The following result will be central.

\begin{theorem} \label{theorem:limit}
Consider a convergent RC described equivalently by \eqref{eq:cl-rc} or \eqref{eq:ol-rc}. Let $U_{f, h}$ be the unique filter induced by \eqref{eq:ol-rc}. If $h$ is uniformly continuous, then for any $u \in (\mathbb{R}^{n})^\mathbb{Z}$, $e \in \mathbb{R}^\mathbb{Z}$ and any $k \in \mathbb{Z}$, 
\begin{equation} \label{eq:nmax}
\begin{split}
y_k & = U_{f, h}(u, e) \rvert_{k} + e_k \\
& = h \circ f(x_{k-1}, u_{k-1}, e_{k-1}) + e_k \\
& = h \circ f( f(x_{k-2}, u_{k-2}, e_{k-2}), u_{k-1}, e_{k-1}) + e_k \\
&  \hspace*{10em} \vdots \\
& = \mathcal{F}(u_{k-1}, u_{k-2}, \ldots, e_{k-1}, e_{k-2}, \ldots) + e_k,
\end{split}
\end{equation}
where the following point-wise limit
\begin{equation} \label{eq:nmax-limit}
\begin{split}
& \hspace*{-0.7em}  U_{f, h}(u, e)\rvert_k = \mathcal{F}(u_{k-1}, u_{k-2}, \ldots, e_{k-1}, e_{k-2}, \ldots ) \\
& \hspace*{-0.7em} \coloneqq \lim_{k_0 \rightarrow -\infty}  h \circ f(\ldots f(f(\xi, u_{k_0}, e_{k_0}), u_{k_0+1}, e_{k_0+1}) \ldots)
\end{split}
\end{equation}
exists and is independent of the initial condition $\xi \in \mathbb{R}^N$. 
\end{theorem}
\begin{proof}
Equation~\eqref{eq:nmax} follows from \eqref{eq:ol-rc}. To show the point-wise limit \eqref{eq:nmax-limit} exists, fix $k \in \mathbb{Z}$ and $\xi \in \mathbb{R}^N$. For any $k_0 \in \mathbb{Z}$ with $k_0 \leq k$, consider a solution $\phi(k; k_0, \xi)$ to \eqref{eq:ol-rc}. 
Then \eqref{eq:nmax-limit} can be written as
$ U_{f, h}(u, e)\rvert_k= \mathcal{F}(u_{k-1}, u_{k-2}, \ldots, e_{k-1}, e_{k-2}, \ldots) = \lim_{k_0 \rightarrow -\infty} h \circ \phi(k; k_0, \xi).$
We show that $\{h \circ \phi(k; k_0, \xi)\}_{k_0 \leq k}$ is a Cauchy sequence and thus \eqref{eq:nmax-limit} exists. Since $h$ is uniformly continuous, it suffices to show $\{ \phi(k; k_0, \xi)\}_{k_0 \leq k}$ is Cauchy.

Let $k_0 \geq k'_0$ and $x^*$ be the reference state solution to  \eqref{eq:ol-rc}. By \eqref{eq:defn-cv} in Definition~\ref{defn:cv}, we have
\begin{equation*}
\begin{split}
& \| \phi(k; k_0, \xi) - \phi(k; k'_0, \xi) \| \\
& \leq \|  x^*_k - \phi(k; k_0, \xi)  \| +  \|  x^*_k - \phi(k; k'_0, \xi)  \| \\
& \leq \beta(\| x^*_{k_0} - \xi \|, k - k_0) + \beta(\| x^*_{k'_0} -\xi \|, k - k'_0) \\
& \leq \beta(R_\xi, k - k_0) + \beta(R_\xi, k - k'_0),
\end{split}
\end{equation*}
where $R_\xi = \max\{ \| x^*_{k_0} - \xi \|, \| x^*_{k'_0} - \xi \| \} < \infty$. Since $\beta \in \mathcal{KL}$, for any $\epsilon > 0$ there exists $k^*_0 \in \mathbb{Z}$ such that for all $k'_0  \leq k_0 \leq k^*_0$, $\beta(R_\xi, k - k_0) + \beta(R_\xi, k - k'_0) < \epsilon$. It follows that $\{ \phi(k; k_0, \xi)\}_{k_0 \leq k}$ is Cauchy for any fixed $k$ and $\xi$.

To show that \eqref{eq:nmax-limit} is independent of $\xi$, let $\phi(k; k_0, \xi')$ be another solution to \eqref{eq:ol-rc} starting at another initial condition $\xi' \neq \xi$. Mimicking the argument above gives
\begin{equation*}
\begin{split}
& \| \phi(k; k_0, \xi) -  \phi(k; k_0, \xi') \| \\
& \leq \beta(\| x^*_{k_0} - \xi \|, k - k_0) + \beta(\| x^*_{k_0} - \xi' \|, k - k_0).
\end{split}
\end{equation*}
The limit \eqref{eq:nmax-limit} is independent of $\xi$ now follows from taking $k_0 \rightarrow - \infty$ and uniform continuity of $h$.
\end{proof}

In Theorem~\ref{theorem:limit}, we have written $x_{k+1} = f(x_k, u_k, e_k)$ as in \eqref{eq:ol-rc}. Equivalently, we can write $x_{k+1} = g(x_k, u_k, y_k)$ as in \eqref{eq:cl-rc}, where $g$ and $f$ are related via \eqref{eq:f}. This leads to the following.
\begin{corollary} \label{coro:limit}
Consider a convergent RC described equivalently by \eqref{eq:cl-rc} or \eqref{eq:ol-rc}.
If $h$ is uniformly continuous, then for any $u \in (\mathbb{R}^n)^\mathbb{Z}$, $e \in \mathbb{R}^\mathbb{Z}$ and any $k \in \mathbb{Z}$,
\begin{equation*}
y_k  = \mathcal{G}(u_{k-1}, u_{k-2}, \ldots, y_{k-1}, y_{k-2}, \ldots) + e_k,
\end{equation*}
where the following point-wise limit 
\begin{equation} \label{eq:narx-limit}
\begin{split}
\hspace*{-0.8em} & \mathcal{G}(u_{k-1}, u_{k-2}, \ldots, y_{k-1}, y_{k-2}, \ldots)   \\
\hspace*{-0.8em}  & \coloneqq \lim_{k_0 \rightarrow -\infty} h \circ g(\ldots g(g(\xi, u_{k_0}, y_{k_0}), u_{k_0+1}, y_{k_0+1}) \ldots )
\end{split}
\end{equation}
exists and is independent of initial condition $\xi \in \mathbb{R}^N$.
\end{corollary}

\subsection{NARX($\infty$) models} 
\label{subsec:narx}
We apply Corollary~\ref{coro:limit} to show that convergent RCs described by \eqref{eq:cl-rc} or \eqref{eq:ol-rc}, such that $f$ is continuous and $h$ is uniformly continuous, implement NARX($\infty$) models defined in \eqref{eq:y-narx} below. By Theorem~\ref{theorem:limit}, such NARX($\infty$) models can also be written as NMAX($\infty$) models defined in \eqref{eq:y-nmax} below.

Let $(\Omega, \Sigma, \mathbb{P})$ be a complete probability space on which all random variables are defined. We say that $\textbf{z}$ is a stochastic process if $\textbf{z}: (\Omega, \Sigma) \rightarrow ((\mathbb{R}^m)^\mathbb{Z}, (\mathcal{R}^m)^\mathbb{Z})$ is measurable, where $(\mathcal{R}^m)^\mathbb{Z}$ is the $\sigma$-algebra generated by cylindrical sets in $(\mathbb{R}^m)^\mathbb{Z}$; see e.g. \cite[Sec.~36]{billingsley2008probability}. 
Further, for any $\omega \in \Omega$, $\textbf{z}(\omega) = \{ \textbf{z}_k(\omega)\}_{k \in \mathbb{Z}}$ is a realization of $\textbf{z}$.

Consider RCs described by \eqref{eq:cl-rc} or \eqref{eq:ol-rc} under $\mathbb{R}^n$-valued and $\mathbb{R}$-valued stochastic processes $\textbf{u}$ and $\textbf{e}$.
The state $\textbf{x}_k$ and the one-step ahead prediction $\hat{\textbf{y}}_k$ are described by a stochastic RC:
\begin{equation*} \label{eq:cl-rc-sto}
\begin{cases}
\textbf{x}_{k+1} = g(\textbf{x}_k, \textbf{u}_k, \textbf{y}_k), \\
\hspace*{1.1em} \hat{\textbf{y}}_k = h(\textbf{x}_k),
\end{cases}
\end{equation*}
where $\textbf{y}_k = \hat{\textbf{y}}_k + \textbf{e}_k$.
Stochasticity of $\mathbf{y}$ arises solely from the stochasticity of $\textbf{u}$ and $\textbf{e}$, and the maps $g, h$ are deterministic. 

For a convergent RC described by \eqref{eq:cl-rc} or \eqref{eq:ol-rc} with uniformly continuous $h$, by Corollary~\ref{coro:limit}, $\textbf{y}_k$ is described by a NARX($\infty$) model such that for all $k \in \mathbb{Z}$,
\begin{equation} \label{eq:y-narx}
\textbf{y}_k = \mathcal{G}(\textbf{u}_{k-1}, \textbf{u}_{k-2}, \ldots, \textbf{y}_{k-1}, \textbf{y}_{k-2}, \ldots) + \textbf{e}_k.
\end{equation}
For each $\omega \in \Omega$, the point-wise limit \eqref{eq:narx-limit} exists and is independent of initial condition $\xi \in \mathbb{R}^N$.

Equivalently, for a convergent RC described by \eqref{eq:cl-rc} or \eqref{eq:ol-rc} with uniformly continuous $h$, by Theorem~\ref{theorem:limit}, $\textbf{y}_k$ is also described by a NMAX($\infty$) model, such that for all $k \in \mathbb{Z}$
\begin{equation} \label{eq:y-nmax}
\begin{split}
\textbf{y}_k & = U_{f,h}(\textbf{u}, \textbf{e})\rvert_k + \textbf{e}_k \\
& =  \mathcal{F}(\textbf{u}_{k-1}, \textbf{u}_{k-2}, \ldots, \textbf{e}_{k-1}, \textbf{e}_{k-2}, \ldots) + \textbf{e}_k.
\end{split}
\end{equation}
For each $\omega \in \Omega$, the point-wise limit \eqref{eq:nmax-limit} 
exists and is independent of initial condition $\xi \in \mathbb{R}^N$.


We now show that if in addition, $f$ defined by \eqref{eq:f} is continuous, then the output $\textbf{y}$ of a NARX($\infty$) model given by \eqref{eq:y-narx} is a well-defined stochastic process. 

\begin{lemma} \label{lemma:narx}
Consider a convergent RC described equivalently by \eqref{eq:cl-rc} or \eqref{eq:ol-rc}. Let $U_{f, h}:((\mathbb{R}^n)^{\mathbb{Z}} \times \mathbb{R}^{\mathbb{Z}}, (\mathcal{R}^n)^{\mathbb{Z}} \times \mathcal{R}^{\mathbb{Z}}) \rightarrow (\mathbb{R}^\mathbb{Z}, \mathcal{R}^\mathbb{Z})$ be the unique filter induced by \eqref{eq:ol-rc}. Suppose that $h: \mathbb{R}^N \rightarrow \mathbb{R}$ is uniformly continuous and $f:\mathbb{R}^N \times \mathbb{R}^n \times \mathbb{R} \rightarrow \mathbb{R}^N$ defined by \eqref{eq:f} is continuous. Then $U_{f, h}$ is measurable and for any stochastic processes $\textbf{u}$ and $\textbf{e}$, the output $\textbf{y}$ of the NARX($\infty$) model defined by \eqref{eq:y-narx} is a stochastic process. 
\end{lemma}
\begin{proof}
Since $\textbf{y}$ is also the output of the corresponding NMAX($\infty$) model \eqref{eq:y-nmax}, it follows that $\textbf{y}$ is a stochastic process if $U_{f, h}$ is measurable. Recall the bijection between $U_{f, h}$ and its functional $F_{f, h}$. This bijection implies that $U_{f, h}$ is measurable if and only if $F_{f, h}$ is measurable \cite[Sec.~II]{gonon2019reservoir}. Hence it suffices to show $F_{f, h} :((\mathbb{R}^n)^{\mathbb{Z}_{-}} \times \mathbb{R}^{\mathbb{Z}_{-}}, (\mathcal{R}^n)^{\mathbb{Z}_{-}} \times \mathcal{R}^{\mathbb{Z}_{-}}) \rightarrow (\mathbb{R}, \mathcal{B}(\mathbb{R})) $ is measurable. By Theorem~\ref{theorem:limit}, for any $u' \in (\mathbb{R}^n)^{\mathbb{Z}_{-}}$ and $e' \in \mathbb{R}^{\mathbb{Z}_{-}}$, $F_{f, h}(u', e') = \lim_{k_0 \rightarrow -\infty}  h \circ f(\ldots f(\xi, u'_{k_0}, e'_{k_0}) \ldots),$ where the limit exists and is independent of $\xi$. Fix a $\xi \in \mathbb{R}^N$, for any $k_0 \in \mathbb{Z}_{-}$, define $F_{f, h}^{k_0}: (\mathbb{R}^n)^{\mathbb{Z}_-} \times \mathbb{R}^{\mathbb{Z}_-} \rightarrow \mathbb{R}$ by 
\begin{equation*}
\begin{split}
& F_{f, h}^{k_0}(u', e') \coloneqq  h \circ f(\ldots f(\xi, u'_{k_0}, e'_{k_0}) \ldots).
\end{split}
\end{equation*}
Then $\lim_{k_o \rightarrow -\infty} F^{k_0}_{f,h}(u', e') = F_{f,h}(u', e')$ point-wise and $F_{f, h}$ is measurable if $F_{f, h}^{k_0}$ is measurable for all $k_0 =\{\ldots, -2, -1\}$ \cite[Theorem~13.4]{billingsley2008probability}. To show this, we write
$$F^{k_0}_{f, h}(u', e') = h\circ \overline{f}_{k_0} (\mathcal{P}_{n}^{k_0}(u'),  \mathcal{P}^{k_0}(e')),$$ where $\mathcal{P}_n^{k_0}\coloneqq \prod_{j=k_0}^{-1} P_n^j: (\mathbb{R}^n)^{\mathbb{Z}_{-}} \rightarrow (\mathbb{R}^n)^{-k_0}$, $\mathcal{P}^{k_0}\coloneqq \prod_{j=k_0}^{-1} P^j: \mathbb{R}^{\mathbb{Z}_{-}} \rightarrow \mathbb{R}^{-k_0}$. Here, $P^j_n(u') = u'_j$ and $ P^j(e')=e'_j$. Furthermore, $\overline{f}_{k_0}: (\mathbb{R}^n)^{-k_0} \times \mathbb{R}^{-k_0} \rightarrow \mathbb{R}^N$ is given by 
\begin{equation*}
\overline{f}_{k_0} (\mathcal{P}_{n}^{k_0}(u'),  \mathcal{P}^{k_0}(e')) =  f(\ldots f(\xi, u'_{k_0}, e'_{k_0}) \ldots).
\end{equation*}
Since $h, \mathcal{P}_n^{k_0}$ and $\mathcal{P}^{k_0}$ are measurable, it remains to show that $\overline{f}_{k_0}$ is measurable. To this end, to simplify notation, for any $i, j \in \mathbb{Z}_{+}$, let $u'_{k_0-j:-1-i}$ and $e'_{k_0-j:-1-i}$ denote the concatenation of $\{u'_{k_0-j}, \ldots, u'_{-1-i}\}$ and $\{e'_{k_0-j}, \ldots, e'_{-1-i}\}$ into a column vector, respectively. We can define $\overline{f}_{k_0}$ recursively via
$\overline{f}_{k_0-1}(u'_{k_0-1:-1}, e'_{k_0-1:-1}) = f(\overline{f}_{k_0}(u'_{k_0-1:-2}, e'_{k_0-1:-2}), u'_{-1}, e'_{-1})$. Using this recursion and continuity of $f$, an inductive argument on $k_0$ shows that $\overline{f}_{k_0}$ is continuous.
Hence $F_{f, h}$ and $U_{f, h}$ are measurable, and $\textbf{y}$ defined by \eqref{eq:y-narx} is a stochastic process.
\end{proof}

\subsection{Stationarity and ergodicity}
\label{subsec:stationary}
In this section, we derive conditions under which $\textbf{y}$ defined by \eqref{eq:y-narx} is stationary and/or ergodic given that $\textbf{u}, \textbf{e}$ are stationary \cite[Definition~2.2]{fan2008nonlinear} and/or ergodic \cite[Sec.~24\&36]{billingsley2008probability}. We equip $\mathbb{R}^m$ with the Borel $\sigma$-algebra $\mathcal{B}(\mathbb{R}^m)$. An $\mathbb{R}$-valued process $\tilde{\textbf{z}}$ is Birkhoff-Khinchin ergodic if $\lim_{L \rightarrow \infty} \frac{1}{L} \sum_{k=0}^{L-1} \tilde{\textbf{z}}_k = \mathbb{E}_{\mathbb{P}}[\tilde{\textbf{z}}_0]$ almost surely (a.s.), where $\mathbb{E}_{\mathbb{P}}[\cdot]$ is the expectation over $\mathbb{P}$. 





To establish statistical properties of the output of a NARX($\infty$) model defined by \eqref{eq:y-narx}, we again exploit the bijection with its associated NMAX($\infty$) model given by \eqref{eq:y-nmax}.
\begin{lemma} \label{lemma:stationary}
Consider a convergent RC described equivalently by \eqref{eq:cl-rc} or \eqref{eq:ol-rc}. Under the assumptions of Lemma~\ref{lemma:narx}, the process $\textbf{y}$ defined by \eqref{eq:y-narx} is stationary (resp. ergodic) if $\textbf{u}$ and $\textbf{e}$ are stationary (resp. ergodic). Furthermore, suppose that $h \circ \textbf{x}^*_0$ and $\textbf{e}_0$ are integrable, where $\textbf{x}^*$ is the reference state solution to \eqref{eq:ol-rc}. Then under the assumptions of Lemma~\ref{lemma:narx}, $\textbf{y}$ is Birkhoff-Khinchin ergodic if $\textbf{u}$, $\textbf{e}$ are stationary and ergodic.
\end{lemma}
\begin{proof}
By \eqref{eq:y-nmax}, we have $\textbf{y}_k = U_{f, h}(\textbf{u}, \textbf{e})\rvert_k + \textbf{e}_k$, where $U_{f, h}$ is the unique filter induced by \eqref{eq:ol-rc}. By Lemma~\ref{lemma:narx}, $U_{f, h}$ is measurable. Now $\textbf{y}$ is stationary (resp. ergodic), given that $\textbf{u}, \textbf{e}$ are stationary (resp. ergodic), follows from \cite[Theorem~36.4]{billingsley2008probability}. To show the second part of the Lemma, by discussions on \cite[p.~526]{billingsley2008probability} and the Birkhoff-Khinchin ergodic theorem \cite[Theorem~24.1]{billingsley2008probability}, it suffices to show that $F_{f, h}(P_n^{\mathbb{Z}_{-}}(\textbf{u}), P^{\mathbb{Z}_{-}}(\textbf{e})) + \textbf{e}_0 : (\Omega, \Sigma, \mathbb{P}) \rightarrow \mathbb{R}$ is integrable, where $F_{f, h}$ is the  functional induced by \eqref{eq:ol-rc}. Recall that
\begin{equation*}
\begin{split}
& \int_{\Omega} \left| F_{f, h}(P_n^{\mathbb{Z}_{-}}(\textbf{u}(\omega)), P^{\mathbb{Z}_{-}}(\textbf{e}(\omega))) + \textbf{e}_0(\omega) \right| \mathbb{P}({\rm d}\omega) \\
& \leq \int_{\Omega} \left(|h(\textbf{x}^*_0(\omega))| + |\textbf{e}_0(\omega)| \right) \mathbb{P}({\rm d}\omega).
\end{split}
\end{equation*}
Integrability of $F_{f, h}(P_n^{\mathbb{Z}_{-}}(\textbf{u}), P^{\mathbb{Z}_{-}}(\textbf{e})) + \textbf{e}_0$ now follows from integrability of $h \circ \textbf{x}^*_0$ and $\textbf{e}_0$.
\end{proof}


\begin{remark} \label{remark:sufficient-integrable}
When $h$ is uniformly continuous (as required in Theorem~\ref{theorem:limit}), a sufficient condition for integrability of $h\circ \textbf{x}^*_0$ is that there exists $0< M < \infty$ such that $\| \textbf{x}^*_{0}\| \leq M$ a.s..
\end{remark}

\begin{remark} \label{remark:one-step-ahead}
A similar argument as in the proof of Lemma~\ref{lemma:stationary} shows that under the assumptions of Lemma~\ref{lemma:stationary}, the one-step ahead prediction $\hat{\textbf{y}}$ and $\hat{\textbf{y}}^2$ are Birkhoff-Khinchin ergodic if $\textbf{u}$, $\textbf{e}$ are stationary and ergodic.
\end{remark}

We conclude this section by imposing the following standard assumptions on $\textbf{u}$ and $\textbf{e}$:
\begin{assumption}\label{assumption1}
$\textbf{u}$ and $\textbf{e}$ are independent, and $\textbf{e}$ is identically and independently distributed (iid). This implies that $\textbf{e}$ is stationary and ergodic \cite[Sec.~36]{billingsley2008probability}.
\end{assumption}

\begin{remark} \label{remark:residual}
Assumption~\ref{assumption1} lays the basis for analyzing the RC residual $\hat{\textbf{e}}_k = \textbf{y}_k - \hat{\textbf{y}}_k$; see \cite{fan2008nonlinear} for further detail. If the RC prediction $\hat{\textbf{y}}$ describes the target data $\textbf{y}$ adequately, $\hat{\textbf{e}}$ should be a proxy for $\textbf{e}$. To test this, we test if $\hat{\textbf{e}}$ are uncorrelated using sample autocorrelation (ACF) via the ``ggAcf'' \texttt{R} command \cite{forecast}, and test if $\textbf{e}$ is Gaussian using the Lilliefors test \cite{lilliefors1967kolmogorov} and Q-Q plot (via the ``lillietest'' and the ``qqplot'' \texttt{Matlab} commands). In system identification, we further test if $\hat{\textbf{e}}$ is independent of $\textbf{u}$ based on their sample cross-correlation via the ``ggCcf'' \texttt{R} command \cite{forecast}.
\end{remark}

\section{Parameter estimation}
\label{sec:para_est}

This section presents parameter estimation of ESNs and QRCs as NARX($\infty$) models. These RCs are described by \eqref{eq:cl-rc} or \eqref{eq:ol-rc}, where $f$ and $h$ satisfy the assumptions in Lemma~\ref{lemma:stationary}. Further, Remark~\ref{remark:sufficient-integrable} holds since ESNs and QRCs admit a compact state-space. By Lemma~\ref{lemma:stationary} and Remark~\ref{remark:one-step-ahead}, if these RCs are convergent, then they realize NARX($\infty$) models whose output $\textbf{y}$ and one-step ahead prediction $\hat{\textbf{y}}$ are stationary and ergodic. Here, $\hat{\textbf{y}}_k = W^\top \overline{h}(\textbf{x}_k) + W_c$, where $W \in \mathbb{R}^N, W_c \in \mathbb{R}$ are output parameters and $\overline{h}$ is not parametrized.


Given time series data $\textbf{y}_k, \textbf{u}_k$ for $k=0, \ldots, L$, the first $L_1$ data are for washing out the effect of the RC's initial condition, the next $L_t = L_2 - L_1$ data are for training and the remaining data are for validation. We apply Theorem~\ref{theorem:cv} to \eqref{eq:ol-rc} to ensure the RC's convergence, the resulting optimization of $W, W_c$ becomes a convex constrained least squares problem,
 $\min_{W, W_c} \frac{1}{L_t}\sum_{k=L_1+1}^{L_2} |\textbf{y}_k - \hat{\textbf{y}}_{k}|^2$, subject to $\ G(W) \leq 0.$
Note that if $\textbf{y}^2$ and $\textbf{y} \hat{\textbf{y}}$ are  Birkhoff-Khinchin ergodic, then parameter estimation is consistent, i.e., as $L_t \rightarrow \infty$, the above cost function becomes 
$
\min_{W, W_c} \mathbb{E}_{\mathbb{P}}[|\textbf{y}_0 - \hat{\textbf{y}}_0|^2]$ a.s..

\subsection{Echo-state networks (ESNs)}
\label{subsec:ESN}
An ESN with state $x_k \in \mathbb{R}^N$ is governed by
\begin{equation} \label{eq:esn}
\begin{cases}
x_{k+1} = \tanh(Ax_k + B u_k + C y_k), \\
\hspace*{1.1em} \hat{y}_{k} = W^\top x_k + W_c
\end{cases}
\end{equation}
where $\tanh(\cdot)$ applies to a vector elementwise. Elements of $A, B, C$ are drawn independently and uniformly from $[-1, 1]$ (for time series modeling, we set $C = 0$). To apply Theorem~\ref{theorem:cv}, we re-express \eqref{eq:esn} in the form of \eqref{eq:ol-rc} by substituting $y_k = \hat{y}_k + e_k = W^\top x_k + W_c + e_k$ into \eqref{eq:esn}:
\begin{equation*} \label{eq:esn-ol}
\begin{split}
x_{k+1} & = f_{\rm ESN}(x_k, u_k, e_k) \\
& \coloneqq \tanh((A + C W^\top)x_k  + C (W_c +e_k)  + B u_k).
\end{split}
\end{equation*}
Since ESNs admit a compact state-space, by Theorem~\ref{theorem:cv}, an ESN is convergent if for any $x_1, x_2 \in \mathbb{R}^N$ and any $k\in \mathbb{Z}$,
\begin{equation*}
\begin{split}
& \| f_{\rm ESN}(x_1, u_k, e_k) - f_{\rm ESN}(x_2, u_k, e_k) \|_P  \\
& \leq \sigma_{\rm m}(A+CW^\top) \| x_1 - x_2 \|_P < \| x_1 - x_2 \|_P ,
\end{split}
\end{equation*}
where $\sigma_{\rm m}(\cdot)$ is the maximum singular value. That is, an ESN is convergent if $\sigma_{\rm m}(A+CW^\top)<1$. We optimize $W,W_c$ for the ESN using YALMIP \cite{Lofberg2004} by passing the non-strict inequality constraint:
\begin{equation} \label{eq:esn-bound}
\begin{bmatrix}
I & A+CW^\top \\ (A+CW^\top)^\top & I 
\end{bmatrix} \geq 10^{-3} I.
\end{equation}

\subsection{Quantum reservoir computers (QRCs)}
\label{subsec:QRC}
Consider an $N$-qubit QRC proposed in \cite{chen2020temporal} described by
\begin{equation} \label{eq:QRC}
\begin{cases}
\rho_{k+1} = (1-\epsilon) T(u_k, y_k) \rho_k + \epsilon \rho_*, \\
\hspace*{1em} \hat{y}_k = \sum_{i=1}^{N} W_i {\rm Tr}(Z^{(i)} \rho_k) + W_c,
\end{cases}
\end{equation}
where $\rho_k$ is a $2^N \times 2^N$ density operator (a positive semidefinite Hermitian matrix with trace one), $\rho_*$ is a fixed density operator whose $(1,1)$-th element is one and all other elements are zero, $Z^{(i)}$ is the Pauli-$Z$ operator on qubit $i$ and $\epsilon = 0.9$. The dynamics $T(u_k, y_k)$ is a completely positive trace-preserving (CPTP) map determined by $u_k$ and $y_k = \hat{y}_k + e_k$. A CPTP map sends a density operator to another density operator \cite{chen2020temporal}. 
We consider
\begin{equation*}
\begin{split}
& T(u_k, y_k) = \frac{1}{n+1} \left( \sum_{j=1}^{n} g(u^{(j)}_k) T_j + g(y_k) T_{n+1}  \right. \\
& \hspace*{6em} \left. + \left(n+1 - \sum_{j=1}^{n} g(u_k^{(j)}) - g(y_k) \right) T_{n+2} \right),
\end{split}
\end{equation*}
where $u^{(j)}_k$ is the $j$-th component of $u_k \in \mathbb{R}^n$. We choose $g(s) = \frac{1}{1 + \exp(-s)}$ with a globally Lipschitz constant $L_g=1/4$. Generally, different $g$'s can be applied to $u^{(j)}_k$ and $y_k$ (for time series modeling, we set $g(y_k)=0$). The CPTP maps $T_l$ for $l = 1, \ldots n+2$ are governed by arbitrary but fixed unitary matrices $U_l$ such that $T_l(\rho_k) = U_l \rho_k U^\dagger_l$, where $\dagger$ denotes the adjoint; see \cite{arxiv} for the details. Such QRCs can be implemented on current quantum machines \cite{chen2020temporal}. 

Let $\| A \|_1 \coloneqq {\rm Tr}(\sqrt{A^\dagger A})$ for any matrix $A$. For any CPTP map $T$, let $\| T \|_{1-1} \coloneqq \sup_{\|A\|_1=1} \| T(A) \|_1$. Since QRCs admit a compact state-space, by \eqref{eq:theorem-cv1} in Theorem~\ref{theorem:cv} (Theorem~\ref{theorem:cv} also applies to $\| \cdot \|_1$), a QRC is convergent if for any density operators $\rho_j (j=1,2)$, any $k \in \mathbb{Z}$, $u \in (\mathbb{R}^n)^\mathbb{Z}$ and $e \in \mathbb{R}^\mathbb{Z}$,
\begin{equation*}
\begin{split}
& \| (1-\epsilon) T(u_k, y_{k, 1}) \rho_1 + \epsilon \rho_* - ((1-\epsilon) T(u_k, y_{k, 2}) \rho_2 + \epsilon \rho_*) \|_1 \\
& \leq \left( (1-\epsilon) + \frac{2 L_g (1-\epsilon)}{n+1} \sum_{i=1}^{N} | W_i | \right)  \| \rho_1 - \rho_2 \|_1 \\
& \leq \theta \| \rho_1 - \rho_2 \|_1,
\end{split}
\end{equation*}
where the first inequality follows from $\| T_{n+1} - T_{n+2} \|_{1-1} \leq 2$ \cite{chen2020temporal} and $|g(y_1) - g(y_2)| \leq L_g | y_1 - y_2| \leq L_g \sum_{i=1}^{N} | W_i | \| \rho_1 - \rho_2 \|_1.$ Hence, a QRC described by \eqref{eq:QRC} is convergent if there exists some $\theta \in (0, 1)$ such that
\begin{equation} \label{eq:qrc-bound}
\sum_{i=1}^{N} |W_i| - \frac{\theta + \epsilon - 1}{1-\epsilon} \frac{n+1}{2 L_g} \leq 0.
\end{equation}
Throughout, we set $\theta = 0.99$. We optimize $W$ using the ``fmincon'' command in \texttt{Matlab}.

We remark that \eqref{eq:esn-bound} and \eqref{eq:qrc-bound} are only sufficient. Nevertheless, numerical experiments suggest that ESNs and QRCs under these constraints are adequate in describing the data. It is a future direction to find relaxation of these constraints.

\section{Numerical examples}
\label{sec:ts}
We employ convergent ESNs and QRCs to model time series. 
For each time series, we randomly generate 50 ESNs with $N=2, \ldots, 10$, and 50 QRCs with qubit number $N=2, \ldots, 5$.
The fitted ESNs and QRCs are selected via final prediction error (${\rm FPE}$) criterion using validation data defined as ${\rm FPE} \coloneqq \left( \frac{L_v+N+1}{L_v-N+1} \right) {\rm MSE}$ and ${\rm MSE} \coloneqq \frac{1}{L_v} \sum_{k=L_2+1}^{L} \hat{\textbf{e}}_k^2$, where $L_v = L - L_2$ and $\textbf{e}_k$ is the residual. We remark that ${\rm FPE}$ assumes $W, W_c$ are unbiased \cite{konishi2008information}. Although this cannot be ensured in general, the decrease in ${\rm MSE}$ is small as $N$ increases and ${\rm FPE}$ prefers the lowest order $N=2$. It is a future research theme to develop model selection methods for the proposed scheme. We employ the root mean-squared error (${\rm RMSE}$) to compare among different models, where ${\rm RMSE} \coloneqq \sqrt{{\rm MSE}}$; see \cite{hyndman2006another} for further discussions. We report the selected RC's ${\rm RMSE}$ and the average ${\rm RMSE}$ of 50 randomly generated RCs with the same dimension $N$ (or qubit number) as the selected RC.

We estimate ${\rm Tr}(Z^{(i)} \rho_j)$ by averaging $M_m$ measurements on quantum machines, whose variance decreases as $1/M_m$ \cite{chen2020temporal}. Here, we assume that the sampling error is negligible by taking $M_m$ sufficiently large. To investigate the effect of decoherence when QRCs are implemented on current quantum machines, we simulate the selected QRCs under dephasing and generalized amplitude damping (GAD) channels as in \cite{IBMQnoise}.

For all detailed numerical settings, see \cite{arxiv}.

\subsection{Nonlinear quantum optics}
\label{subsec:kerr}
This example demonstrates that RCs can act as nonlinear Wiener filters to extract the signal  component of a highly noisy time series from nonlinear quantum optics.
We consider a low-photon-number Kerr cavity with two input-output ports (Fig.~\ref{fig:kerr-cavity}(a) or \cite[Fig.~2(a)]{santori2014quantum}), whose internal mode is governed by a Hamiltonian $H_0 = \Delta a^\dagger a + \chi (a^\dagger)^2 a^2$, where $a$ is the annihilation operator, $\Delta=100$ is the detuning from a reference frequency, and $\chi=-5$ governs the nonlinearity. The cavity is coupled to two incoming fields $\alpha_{{\rm in}j}(t)$ via operators $L_j = -\iota \sqrt{\kappa_j} a$, where $\iota=\sqrt{-1}$ and $\kappa_j = 25$ for $j=1,2$. Here, $\alpha_{\rm in1}(t)$ is a coherent field with a constant amplitude $\eta=21.5$ and $\alpha_{\rm in2}(t)$ is in the vacuum state.

\begin{figure}[ht]
\centering
\includegraphics[scale=0.8]{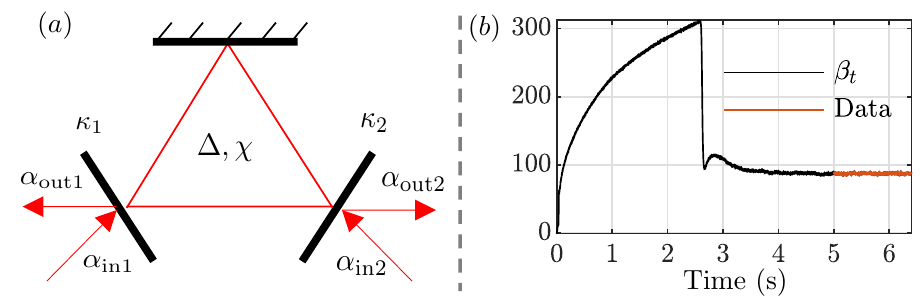}
\caption{(a) Kerr-nonlinear cavity with two input-output ports. The top mirror is fully reflective (without loss) and the other mirrors are partially transmitting. (b) The simulated trajectory of $\beta_k$ and the data employed.}
\label{fig:kerr-cavity}
\end{figure}

\begin{figure}[ht]
\centering
\includegraphics[scale=0.8, trim={2.5em, 0em, 0em, 0em}, clip]{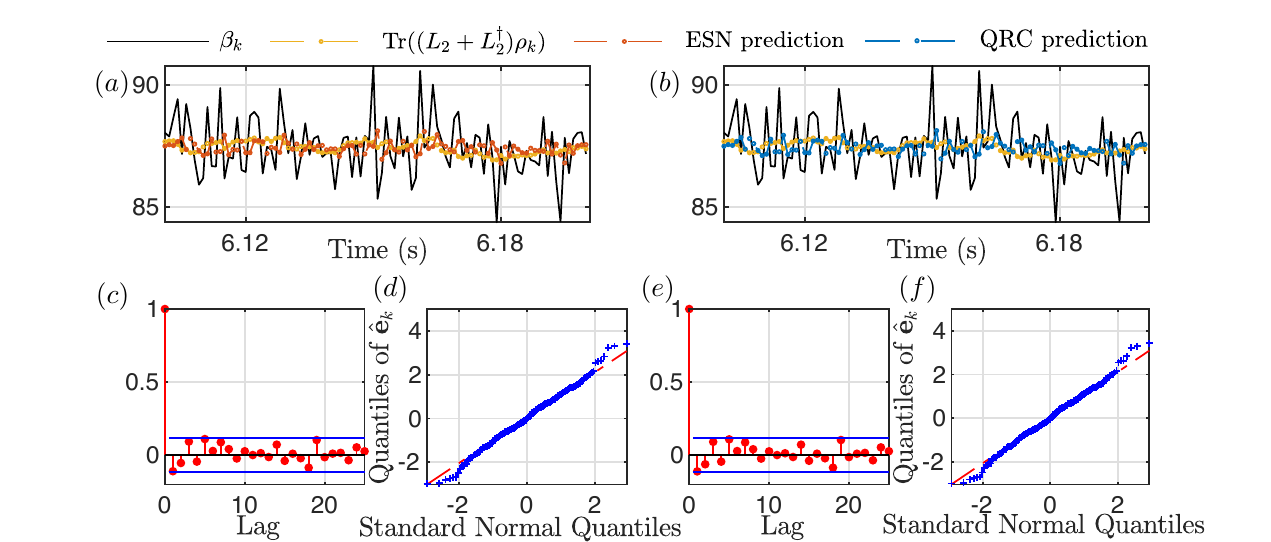}
\caption{The noisy time series $\beta_k$, the signal $\alpha_k = {\rm Tr}((L_2 + L_2^\dagger)\rho_k)$, (a) the ESN prediction and (b) the QRC prediction on the first 100 validation data. (c) The ESN residual sample ACF (horizontal blue lines show the 95\% CI). (d) The ESN residual Q-Q plot. (e) The QRC residual sample ACF. (e) The QRC residual Q-Q plot.}
\label{fig:kerr}
\end{figure}

\begin{table}[!ht]
\centering
\caption{Dimension (or qubit number) $N$, $p$-value, ${\rm RMSE}$ and average ${\rm RMSE}$ ($\overline{{\rm RMSE}}$) of RCs for stochastic modeling.}
\label{table1}
\scalebox{1}{
\begin{tabular}{c|c|c|c|c|c|c}
& \multicolumn{2}{c|}{Optics}&   \multicolumn{2}{c|}{Meteorology} & \multicolumn{2}{c}{Coupled electric drive} \\
\hline\hline 
    &ESN  & QRC &  ESN & QRC & ESN & QRC\\
\hline
$N$ & 2 & 2 & 2 & 2 & 2 & 2 \\
$p$-value &0.5  & 0.5 &  0.5 & 0.23 &  0.5 & 0.44 \\
${\rm RMSE}$  & 1.12  & 1.12 & 0.047 & 0.049 & 0.10 & 0.11\\
$\overline{{\rm RMSE}}$ & 1.12 & 1.12 &  0.055 & 0.058 & 0.12 & 0.14\\
\end{tabular}
}
\end{table}

We obtain a discretized trajectory $\beta_k = {\rm Tr}((L_2 + L_2^\dagger)\rho_k) + \eta_k$ from homodyne measurement on the second output port with sampling time $10^{-3}$ on Qutip \cite{johansson2012qutip}. Here, $\rho_k$ is the state at time $k$ and $\eta_k$ is a quantum Gaussian white noise. Our goal is to employ RCs to separate the ``signal part'' $\alpha_k = {\rm Tr}((L_2 + L_2^\dagger)\rho_k)$ from the highly noisy $\beta_k$. We simulate $\beta_k$ starting in the vacuum state on a truncated Hilbert space of dimension $1000$ for $6.4$s and employ the data after $5$s until $\rho_k$ reaches a steady state, see Fig.~\ref{fig:kerr-cavity}(b). 

We set $L_w = 99$, $L_t = 1000$ and $L_v = 300$. ESN and QRC with $N=2$ achieve comparable ${\rm RMSE}$ and average ${\rm RMSE}$, their residual sample ACFs show no autocorrelation within 95\% CI and pass the Lilliefors test; see Table~\ref{table1} and Fig.~\ref{fig:kerr}. We report $0.5$ for p-values$\geq 0.5$. 
In Fig.~\ref{fig:kerr}(a)(b), we observe that the ESN and QRC predictions follow the signal $\alpha_k$ with a normalized root mean-squared error between them of $\sqrt{ \frac{\sum_{k=L_2+1}^{L} (\alpha_k - \hat{\textbf{y}}_k)^2}{\sum_{k=L_2+1}^{L} \alpha^2_k}}= 0.0047$ and $0.0048$, respectively. These suggest that QRC can act as a nonlinear Wiener filter, separating the signal $\alpha_k$ from the noisy time series $\beta_k$.

For all decoherence strengths, QRC obtains ${\rm RMSE}=1.12$ as the noiseless QRC, and its residuals show no autocorrelation and pass the Lilliefors test with $p$-value$\geq 0.5$. QRC under decoherence can still effectively extract $\alpha_k$ from $\beta_k$.

\begin{figure}[!ht]
\centering
\includegraphics[scale=0.8, trim={2.5em 0em 0.5em 0em}, clip]{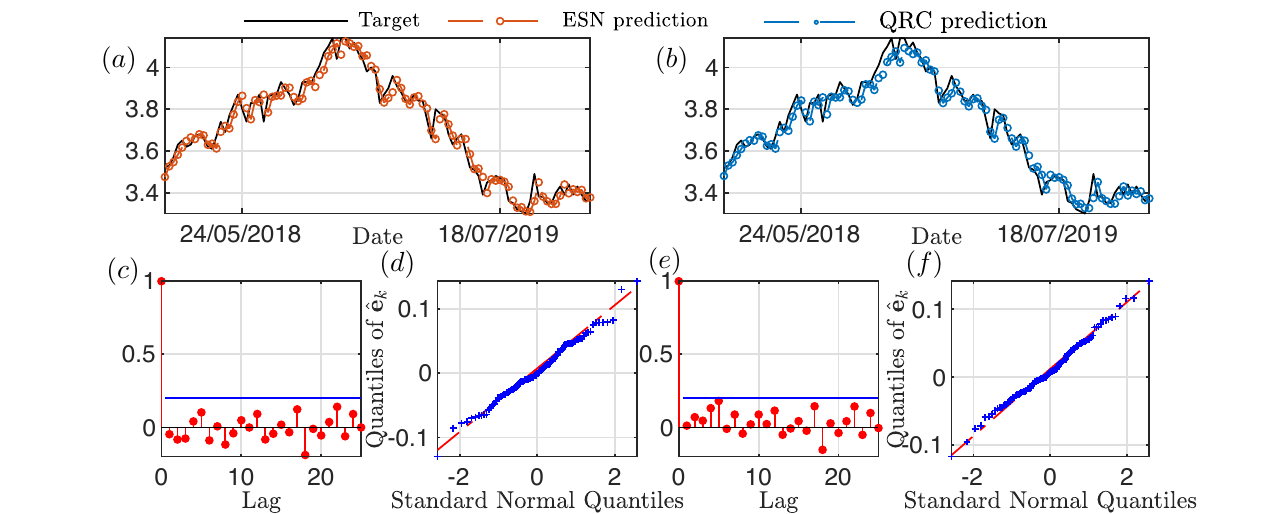}
\caption{The finance time series, (a) the ESN prediction and (b) the QRC prediction on validation data. (c) The ESN residual sample ACF. (d) The ESN residual Q-Q plot. (e) The QRC residual sample ACF. (f) The QRC residual Q-Q plot.}
\label{fig:finance}
\end{figure}

\subsection{Finance}
\label{subsec:mortgage}
This time series describes weekly 5/1-year adjustable rate mortgage average (2005-20) in the US \cite{mortgage}. After removing trend and seasonal components using the ``mstl'' \texttt{R} command \cite{forecast}, the data is highly correlated up to 300 lags. This example tests RCs' ability to model highly correlated data. We set $L_w=100$, $L_t=580$ and $L_v=100$. All RCs achieve comparable ${\rm RMSE}$. The QRC average ${\rm RMSE}$ is similar to the selected QRC's ${\rm RMSE}$, whereas the difference for ESN is more pronounced. All RCs achieve uncorrelated residuals (within 95\% CI) and pass the Lillefors test; see Table~\ref{table1} and Fig.~\ref{fig:finance}. These suggest that both RCs are capable of modeling this highly correlated time series.

QRC under GAD experiences an increased {\rm RMSE}. Despite this, QRC residuals under decoherence show no autocorrelation within 95\% CI and pass the Lilliefors test with $p$-value$\geq 0.35$. This suggests that decoherence does not significantly impact QRC's ability to model highly correlated data.

\subsection{Coupled electric drive system}
\label{sec:sys-id}
We employ RCs on modeling a single-input (i.e., $u_k \in \mathbb{R}$) single-output nonlinear system consists of two electric motors driving a pulley using a flexible belt \cite{pulley}. Input data is a pseudo-random binary sequence (persistently exciting) with amplitude $0.5$ and only $L=500$   data are available, presenting a challenge for RCs. 

We exploit spatial multiplexing, see \cite{chen2020temporal}, where outputs of two distinct and non-interacting RC members are combined linearly. The first member processes both $\textbf{y}_k$ and $\textbf{u}_k$ as described in Sec.~\ref{sec:para_est}, and the second only processes $\textbf{u}_k$. We label each member's parameters with subscripts $1,2$. For simplicity, we set $N=N_1=N_2$ so that dimensions of both members (or numbers of qubits) are the same, with $N=2$ preferred by FPE for all RCs. For the second ESN member, we set $C_2=0$ and $\sigma_{\rm m}(A_2)=0.7$, where the latter ensures convergence. For the multiplexed QRC, we set $\epsilon_1=0.5$ and $\epsilon_2=0.9$, and set the second QRC's CPTP map $T^{(2)}(u_k)$ as
\begin{equation*}
\begin{split}
T^{(2)}(u_k) =  g(u_k) T_1^{(2)}  + (1 -  g(u_k)) T^{(2)}_{2} ,
\end{split}
\end{equation*}
where $T^{(2)}_l (\rho_k) = U^{(2)}_l \rho_k (U^{(2)}_l)^\dagger$ for some arbitrary but fixed unitaries $U^{(2)}_l$ for $l=1, 2$; see \cite{arxiv} for the details. The second QRC is convergent by construction. For multiplexed RCs, output parameters of their second member are not constrained. 


\begin{figure}[!ht]
\centering
\includegraphics[scale=0.8, trim={2.5em, 0, 0, 0}, clip]{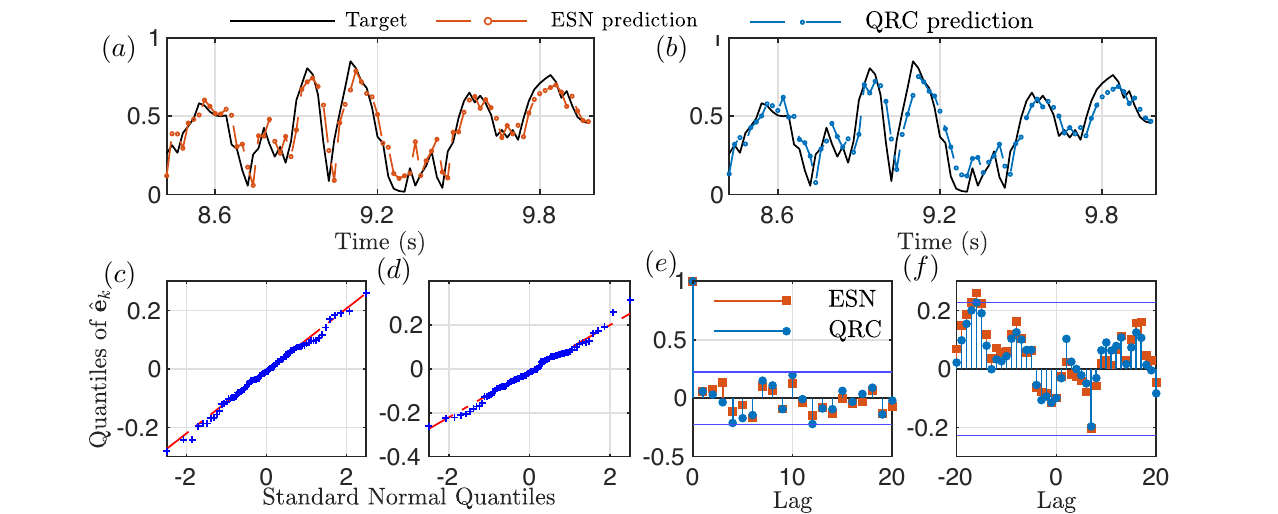}
\caption{The target output, (a) the ESN prediction and (b) the QRC prediction on validation data. (c) The ESN residual Q-Q plot. (d) The QRC residual Q-Q plot. (e) The ESN (blue) and QRC (red) sample ACF. (f) The ESN (blue) and QRC (red) sample cross-correlation between inputs and residuals.}
\label{fig:pulley}
\end{figure}

We set $L_w = 20$, $L_t = 400$ and $L_v = 79$. All RCs achieve comparable ${\rm RMSE}$ and average ${\rm RMSE}$, and they pass the Lilliefors test. Residuals of ESN and QRC are uncorrelated and independent of inputs (within 95\% CI); see Table~\ref{table1} and Fig.~\ref{fig:pulley}. Under both decoherence channels, residuals of QRC show autocorrelation at lag 1. Despite this, the increase in the QRC ${\rm RMSE}$ is small, by at most $0.01$ compared to the noiseless QRC.

\section{Conclusion} 
\label{sec:conclu}
We have introduced convergent reservoir computers with output feedback as stationary and ergodic NARX($\infty$) models. Our approach can harness nonlinear dynamical systems for temporal information processing, making them  versatile for nonlinear stochastic modeling. 
Numerical experiments demonstrate that these reservoir computers with a few tunable parameters are adequate in modeling diverse data sets.

Many exciting problems remain open for future research, such as improving the modeling performance through reservoir design and developing model selection methods for the proposed scheme. With this view, this work opens the potential for reservoir computing paradigm to tackle traditional challenges encountered in control and time series modeling, further bridging these scientific disciplines.


\bibliographystyle{IEEEtran}
\bibliography{ts_brief}

\begin{thebibliography}{10}
\providecommand{\url}[1]{#1}
\csname url@samestyle\endcsname
\providecommand{\newblock}{\relax}
\providecommand{\bibinfo}[2]{#2}
\providecommand{\BIBentrySTDinterwordspacing}{\spaceskip=0pt\relax}
\providecommand{\BIBentryALTinterwordstretchfactor}{4}
\providecommand{\BIBentryALTinterwordspacing}{\spaceskip=\fontdimen2\font plus
\BIBentryALTinterwordstretchfactor\fontdimen3\font minus
  \fontdimen4\font\relax}
\providecommand{\BIBforeignlanguage}[2]{{%
\expandafter\ifx\csname l@#1\endcsname\relax
\typeout{** WARNING: IEEEtran.bst: No hyphenation pattern has been}%
\typeout{** loaded for the language `#1'. Using the pattern for}%
\typeout{** the default language instead.}%
\else
\language=\csname l@#1\endcsname
\fi
#2}}
\providecommand{\BIBdecl}{\relax}
\BIBdecl

\bibitem{fan2008nonlinear}
J.~Fan and Q.~Yao, \emph{Nonlinear time series: nonparametric and parametric
  methods}.\hskip 1em plus 0.5em minus 0.4em\relax Springer Science \& Business
  Media, 2008.

\bibitem{ljung2010perspectives}
L.~Ljung, ``Perspectives on system identification,'' \emph{Annual Reviews in
  Control}, vol.~34, no.~1, pp. 1--12, 2010.

\bibitem{boyd1985fading}
S.~Boyd and L.~Chua, ``Fading memory and the problem of approximating nonlinear
  operators with volterra series,'' \emph{IEEE Transactions on circuits and
  systems}, vol.~32, no.~11, pp. 1150--1161, 1985.

\bibitem{kumpati1990identification}
S.~N. Kumpati \emph{et~al.}, ``Identification and control of dynamical systems
  using neural networks,'' \emph{IEEE Transactions on neural networks}, vol.~1,
  no.~1, pp. 4--27, 1990.

\bibitem{leontaritis1985input}
I.~Leontaritis and S.~A. Billings, ``Input-output parametric models for
  non-linear systems part ii: stochastic non-linear systems,''
  \emph{International journal of control}, vol.~41, no.~2, pp. 329--344, 1985.

\bibitem{schoukens2017identification}
M.~Schoukens and K.~Tiels, ``Identification of block-oriented nonlinear systems
  starting from linear approximations: A survey,'' \emph{Automatica}, vol.~85,
  pp. 272--292, 2017.

\bibitem{tanaka2019recent}
G.~Tanaka \emph{et~al.}, ``Recent advances in physical reservoir computing: A
  review,'' \emph{Neural Networks}, vol. 115, pp. 100--123, 2019.

\bibitem{nakajima2021reservoir}
K.~Nakajima and I.~Fischer, \emph{Reservoir Computing: Theory, Physical
  Implementations, and Applications}.\hskip 1em plus 0.5em minus 0.4em\relax
  Springer Singapore, 2021.

\bibitem{billingsley2008probability}
P.~Billingsley, \emph{Probability and measure}.\hskip 1em plus 0.5em minus
  0.4em\relax John Wiley \& Sons, 2008.

\bibitem{pathak2018model}
J.~Pathak \emph{et~al.}, ``Model-free prediction of large spatiotemporally
  chaotic systems from data: A reservoir computing approach,'' \emph{Physical
  review letters}, vol. 120, no.~2, p. 024102, 2018.

\bibitem{kim2020time}
T.~Kim and B.~R. King, ``Time series prediction using deep echo state
  networks,'' \emph{Neural Computing and Applications}, vol.~32, no.~23, pp.
  17\,769--17\,787, 2020.

\bibitem{fujii2017harnessing}
K.~Fujii and K.~Nakajima, ``Harnessing disordered-ensemble quantum dynamics for
  machine learning,'' \emph{Physical Review Applied}, vol.~8, no.~2, p. 024030,
  2017.

\bibitem{chen2019towards}
J.~Chen, H.~I. Nurdin, and N.~Yamamoto, ``Towards single-input single-output
  nonlinear system identification and signal processing on near-term quantum
  computers,'' in \emph{2019 IEEE 58th Conference on Decision and Control
  (CDC)}.\hskip 1em plus 0.5em minus 0.4em\relax IEEE, 2019, pp. 401--406.

\bibitem{chen2020temporal}
------, ``Temporal information processing on noisy quantum computers,''
  \emph{Phys. Rev. Applied}, vol.~14, p. 024065, Aug 2020.

\bibitem{markovic2020quantum}
D.~Markovi{\'c} and J.~Grollier, ``Quantum neuromorphic computing,''
  \emph{Applied Physics Letters}, vol. 117, no.~15, p. 150501, 2020.

\bibitem{pavlov2008convergent}
A.~Pavlov and N.~van~de Wouw, ``Convergent discrete-time nonlinear systems: the
  case of {PWA} systems,'' in \emph{2008 American Control Conference}.\hskip
  1em plus 0.5em minus 0.4em\relax IEEE, 2008, pp. 3452--3457.

\bibitem{tran2018convergence}
D.~N. Tran, B.~S. R{\"u}ffer, and C.~M. Kellett, ``Convergence properties for
  discrete-time nonlinear systems,'' \emph{IEEE Transactions on Automatic
  Control}, vol.~64, no.~8, pp. 3415--3422, 2018.

\bibitem{grigoryeva2018echo}
L.~Grigoryeva and J.-P. Ortega, ``Echo state networks are universal,''
  \emph{Neural Networks}, vol. 108, pp. 495--508, 2018.

\bibitem{grigoryeva2016reservoir}
L.~Grigoryeva, J.~Henriques, and J.-P. Ortega, ``Reservoir computing:
  information processing of stationary signals,'' in \emph{Joint 2016 CSE, EUC
  and DCABES}.\hskip 1em plus 0.5em minus 0.4em\relax IEEE, 2016, pp. 496--503.

\bibitem{gonon2019reservoir}
L.~Gonon and J.-P. Ortega, ``Reservoir computing universality with stochastic
  inputs,'' \emph{IEEE transactions on neural networks and learning systems},
  vol.~31, no.~1, pp. 100--112, 2019.

\bibitem{bollt2021explaining}
E.~Bollt, ``{On explaining the surprising success of reservoir computing
  forecaster of chaos? The universal machine learning dynamical system with
  contrast to VAR and DMD},'' \emph{Chaos: An Interdisciplinary Journal of
  Nonlinear Science}, vol.~31, no.~1, p. 013108, 2021.

\bibitem{lin1996smooth}
Y.~Lin, E.~D. Sontag, and Y.~Wang, ``A smooth converse {Lyapunov} theorem for
  robust stability,'' \emph{SIAM Journal on Control and Optimization}, vol.~34,
  no.~1, pp. 124--160, 1996.

\bibitem{chen2019learning}
J.~Chen and H.~I. Nurdin, ``Learning nonlinear input--output maps with
  dissipative quantum systems,'' \emph{Quantum Information Processing},
  vol.~18, no.~7, p. 198, 2019.

\bibitem{forecast}
\BIBentryALTinterwordspacing
R.~Hyndman \emph{et~al.}, \emph{{{forecast}: Forecasting functions for time
  series and linear models}}, 2020, {R package version 8.12}. [Online].
  Available: \url{http://pkg.robjhyndman.com/forecast}
\BIBentrySTDinterwordspacing

\bibitem{lilliefors1967kolmogorov}
H.~W. Lilliefors, ``{On the Kolmogorov-Smirnov test for normality with mean and
  variance unknown},'' \emph{Journal of the American statistical Association},
  vol.~62, no. 318, pp. 399--402, 1967.

\bibitem{Lofberg2004}
J.~L{\"{o}}fberg, ``{YALMIP : A Toolbox for Modeling and Optimization in
  MATLAB},'' in \emph{In Proceedings of the CACSD Conference}, Taipei, Taiwan,
  2004.

\bibitem{arxiv}
J.~Chen and H.~I. Nurdin, ``Nonlinear autoregression with convergent dynamics
  on novel computational platforms,'' \emph{arXiv preprint}, 2021.

\bibitem{konishi2008information}
S.~Konishi and G.~Kitagawa, \emph{Information criteria and statistical
  modeling}.\hskip 1em plus 0.5em minus 0.4em\relax Springer Science \&
  Business Media, 2008.

\bibitem{hyndman2006another}
R.~J. Hyndman and A.~B. Koehler, ``Another look at measures of forecast
  accuracy,'' \emph{International journal of forecasting}, vol.~22, no.~4, pp.
  679--688, 2006.

\bibitem{IBMQnoise}
\BIBentryALTinterwordspacing
IBM, \emph{{Device backend noise model simulations}}, 2021. [Online].
  Available: \url{https://qiskit.org/documentation/tutorials/simulators/}
\BIBentrySTDinterwordspacing

\bibitem{santori2014quantum}
C.~Santori \emph{et~al.}, ``Quantum noise in large-scale coherent nonlinear
  photonic circuits,'' \emph{Physical Review Applied}, vol.~1, no.~5, p.
  054005, 2014.

\bibitem{johansson2012qutip}
J.~R. Johansson, P.~D. Nation, and F.~Nori, ``{QuTiP: An open-source Python
  framework for the dynamics of open quantum systems},'' \emph{Computer Physics
  Communications}, vol. 183, no.~8, pp. 1760--1772, 2012.

\bibitem{mortgage}
\BIBentryALTinterwordspacing
\emph{5/1-Year adjustable rate mortgage average in the {United States}}, 2020.
  [Online]. Available: \url{https://fred.stlouisfed.org/series/MORTGAGE5US}
\BIBentrySTDinterwordspacing

\bibitem{pulley}
T.~Wigren and M.~Schoukens, ``Coupled electric drives data set and reference
  models,'' \emph{Technical Report, Department of Information Technology,
  Uppsala University}, 2017.

\end{thebibliography}

\end{document}